\documentclass[abstract=true, DIV=11, a4paper]{scrartcl}

\usepackage{algorithm}
\usepackage[noEnd,commentColor=black]{algpseudocodex}
\usepackage{amsmath}
\usepackage{amssymb}
\usepackage{amsthm}
\usepackage[mathscr]{euscript}
\usepackage[shortlabels]{enumitem}
\usepackage{graphicx} %
\usepackage{subcaption}
\usepackage{thmtools}
\usepackage{xcolor}
\usepackage{microtype}

\usepackage[normalem]{ulem}

\usepackage{tcolorbox}

\usepackage[colorlinks]{hyperref}
\usepackage[capitalize]{cleveref}

\crefname{equation}{eq.}{eqs.} %
\crefname{enumi}{}{} %
\crefname{icase}{case}{cases}
\crefname{ipart}{part}{parts}
\crefname{iprop}{property}{properties}
\crefname{iinv}{invariant}{invariants}

\crefformat{section}{\S\,#2#1#3}
\crefformat{subsection}{\S\,#2#1#3}
\crefrangeformat{section}{\S\S\,#3#1#4--#5#2#6}
\crefmultiformat{section}{\S\S\,#2#1#3}{,~#2#1#3}{,~#2#1#3}{,~#2#1#3}

\usepackage{authblk}
\usepackage[normalem]{ulem}

\usepackage{tikz}
\usetikzlibrary{backgrounds}
\usetikzlibrary{calc}

\newcommand{\N}{\mathbb{N}}

\newcommand{\fO}{\mathcal{O}}

\newcommand{\Algo}[1]{\mathbb{#1}}

\newcommand{\ceil}[1]{\lceil #1 \rceil}
\newcommand{\floor}[1]{\lfloor #1 \rfloor}

\newcommand{\connected}[1]{\def\temp{#1}\ifx\temp\empty\sim\else\overset{#1}{\sim}\fi}

\DeclareMathOperator{\Val}{Val} %

\DeclareMathOperator{\DCC}{DCC} %

\renewcommand{\complement}{\mathrm{c}}

\tikzset{
	point/.style={circle, fill, inner sep=1.5pt},
	smallpoint/.style={point, inner sep=1.2pt},
	tinypoint/.style={point, inner sep=1pt},
	hlbox/.style={fill, {white!90!black}},
	subrect/.style={draw, fill={white!80!cyan}}, %
	msrect/.style=subrect %
}

\newtheorem{theorem}{Theorem}[section]
\newtheorem{lemma}[theorem]{Lemma}
\newtheorem{proposition}[theorem]{Proposition}
\newtheorem{conjecture}[theorem]{Conjecture}

\newtheorem{observation}[theorem]{Observation}

\newcommand{\PPM}{\textsc{PPM}}
\newcommand{\SeqPPM}{\textsc{SequencePPM}}

\title{Permutation patterns in streams}
\author{Benjamin Aram Berendsohn}
\affil{Max Planck Institute for Informatics}
\date{\vspace{-5ex}}

\begin{document}

\maketitle

\begin{abstract}
Permutation patterns and pattern avoidance are central, well-studied concepts in combinatorics and computer science. Given two permutations $\tau$ and $\pi$, the \emph{pattern matching} problem (PPM) asks whether $\tau$ contains $\pi$.
This problem arises in various contexts in computer science and statistics and has been studied extensively in exact-, parameterized-, approximate-, property-testing- and other formulations.   

In this paper, we study pattern matching in a \emph{streaming setting}, when the input $\tau$ is revealed sequentially, one element at a time. There is extensive work on the space complexity of various statistics in streams of integers. The novelty of our setting is that the input stream is \emph{a permutation}, which allows inferring some information about future inputs. Our algorithms crucially take advantage of this fact, while existing lower bound techniques become difficult to apply. 

We show that the complexity of the problem changes dramatically depending on the pattern~$\pi$. The space requirement is:
\begin{itemize}
\item $\Theta(k\log{n})$ for the monotone patterns $\pi = 12\dots k$, or $\pi = k\dots21$, 
\item $\fO(\sqrt{n\log{n}})$ for $\pi \in \{312,132\}$,
\item $\fO(\sqrt{n} \log n)$ for $\pi \in \{231,213\}$,
\item $\widetilde{\Theta}_{\pi}(n)$ for all other $\pi$.
\end{itemize}
If $\tau$ is an arbitrary sequence of integers (not necessary a permutation), we show that the complexity is $\widetilde{\Theta}_{\pi}(n)$ in all except the first (monotone) cases.
\end{abstract}

\section{Introduction}
Given a permutation $\tau$ of $[n]$, the ``text'', and a permutation $\pi$ of $[k]$, the ``pattern'',
we say that $\tau$ \emph{contains} $\pi$ if there are indices $i_1 < \cdots < i_k$, such that $\tau(i_j) < \tau(i_\ell)$ if and only if $\pi(j) < \pi(\ell)$, for all $j,\ell \in [k]$, i.e., the subsequence $\tau(i_1) \dots \tau(i_k)$ of $\tau$ is \emph{order-isomorphic} to $\pi$. Otherwise we say that $\tau$ \emph{avoids} $\pi$. Notice that the subsequence in question need not appear contiguously in $\tau$.

Permutation patterns are a central concept in combinatorics, whose origins lie partly in computer science, e.g.,  Knuth~\cite{Knuth1968} characterized stack-sortable permutations as those that avoid the pattern 231.\footnote{As common in the literature, we often use a compact notation for patterns, i.e., $231$ instead of $(2,3,1)$.} For an entry point on the extensive literature on permutation patterns, we refer to textbooks~\cite{ kitaev2011patterns,bona2022combinatorics}, a survey~\cite{vatter2014permutation}, and the yearly conference Permutation Patterns~\cite{patt_conf}.

From an algorithmic point of view, perhaps the most natural question is \emph{permutation patterm matching}~(PPM). This problem asks, given $\tau$ and $\pi$, whether $\tau$ contains $\pi$. The study of the complexity of the problem was initiated by Bose, Buss, and Lubiw~\cite{BoseEtAl1998} who showed it to be NP-complete. A~breakthrough result of Guillemot and Marx showed that the problem is fixed-parameter-tractable in terms of the pattern length $k$, i.e., can be solved in running time of the form $f(k)\cdot n$. This has led to the development of the concept of twin-width, with important generalizations from permutations to graphs, and to other structures~\cite{bonnet2021twin, bonnet2021twinperm}.
For larger patterns, the trivial algorithms with running time $\fO(2^n)$ or $\fO(n^k)$ have been improved to $O(1.414^n)$~\cite{GawrychowskiRzepecki2022PPM}, resp., $n^{0.25k + o(1)}$~\cite{BerendsohnEtAl2021CSP}.

The problem has also been extensively studied for various special cases, e.g., for patterns of bounded treewidth~\cite{AhalRabinovich2008,BerendsohnEtAl2021CSP,JelinekOplerEtAl2024PPMMerge}.

Until now, most studies of the PPM problem assumed that we can, on demand, access arbitrary elements of the input permutation. Given the inherently sequential nature of the problem and the possibility of extremely large inputs, it is natural to assume that the elements of the permutation arrive one-by-one, and that the algorithm cannot revisit earlier elements  (unless explicitly storing them). %

This is known as the \emph{streaming} model, and we initiate the study of the PPM problem in this model.
As typical in streaming, the task is to solve PPM with as little space as possible (running time is a secondary concern).
In this paper, we study deterministic algorithms for the streaming version of PPM.
We show that $\widetilde{\Theta}_\pi(n)$ bits of space\footnote{$\widetilde{\fO}$, $\widetilde{\Theta}$, and $\widetilde{\Omega}$ hide polylogarithmic factors, and the subscript $\pi$ hides constant factors depending on $\pi$.} are required for almost all patterns $\pi$.
The only exception are monotone patterns, where $\Theta( k \log n)$ bits of space suffice (using the textbook algorithm for \emph{longest increasing subsequence}), and non-monotone patterns of size three (312, 132, 213, and 231), where we show that the space complexity lies between $\Omega(\log n)$ and $\widetilde{\fO}(\sqrt{n})$.

The definition of pattern containment easily extends to the case where $\tau$ is an arbitrary sequence of distinct comparable elements (typically, integers). In many settings, this makes no difference, since the elements of $\tau$ can be reduced to an order-isomorphic permutation of~$[n]$ by sorting them using $\Theta(n \log n)$ time and space. In the streaming setting, it turns out that even this moderate relaxation has a large effect. We show that PPM requires $\Omega(n)$ bits of space even for non-monotone patterns of size three, when the input $\tau$ is a sequence of distinct integers from $[n]$. Contrast this with the $\widetilde{\fO}(\sqrt{n})$ upper bound when $\tau$ is a permutation of $[n]$.

This stark separation between the two models can be explained by the fact that when the elements of $\tau$ are known to be exactly $[n]$, both the presence and the non-presence of a value at a certain location in the stream gives us some information. Our $\widetilde{\fO}(\sqrt{n})$-space algorithms crucially exploit this idea.
On the other hand, proving lower bounds in this model is hard: Reductions from standard \emph{communication complexity} problems are complicated by the fact that each half of the input is strongly correlated with the other half. We overcome this problem for non-monotone patterns of size at least four with several individual reductions, but we were not able to prove any non-trivial lower bound for any of the patterns 312, 132, 213, 231, and we leave this as an intriguing open question.

We now review some related work on algorithmic problems tied to permutation patterns. Then, we state our results formally.

\paragraph{Related work.}

The sequence variant of PPM has been considered in the non-deterministic model by Jelínek, Opler, and Valtr~\cite{JelinekOplerEtAl2024PPMMerge}, who applied it to the problem of recognizing \emph{merges} of permutations.

Another important recent line of work that is close to ours studies PPM in a \emph{property-testing} framework. Here, the input sequence is assumed to either avoid the pattern $\pi$, or be \emph{$\varepsilon$-far} from avoiding $\pi$, i.e., differ in $\varepsilon n$ elements from any $\pi$-avoiding sequence. The task is to solve PPM by querying only a \emph{sublinear} number of elements (with random access). Both adaptive and nonadaptive variants have been studied~\cite{NewmanRabinovichEtAl2019PropTesting,Ben-EliezerCanonne2018PPMPropTesting,Ben-EliezerLetzterEtAl2022PPMPropTesting,NewmanVarma2024PropTesting,Zhang2024PPMPropTesting}, and other distance models were also considered~\cite{FoxW18}. The special case of $\pi = 21$ is also known as \emph{monotonicity testing}~(see, e.g., Fischer~\cite{Fischer2004PropTesting}).

Another, more difficult problem related to PPM is pattern \emph{counting}, where the number of occurrences of $\pi$ in $\tau$ must be determined. In the usual non-streaming setting, fixed-parameter tractability, or even a runtime of the form $n^{o(k/\log{k})}$ would contradict standard complexity assumptions~\cite{GuillemotMarx2014,BerendsohnEtAl2021CSP} but the above exponential runtimes for PPM largely transfer.

Counting smaller concrete patterns is important in non-parametric statistics: Even-Zohar and Leng~\cite{Even-ZoharLeng2021PPM} observed that Kendall's $\tau$, Spearman's $\rho$, the Bergsma–Dassios–Yanagimoto test, and Hoeffding’s independence test can be inferred from counting patterns of size 2, 3, 4, and 5, respectively. Accordingly, counting concrete small patterns has been the subject of study in both the exact and approximate sense~\cite{Even-ZoharLeng2021PPM,DudekGawrychowski2020PPM4, approx_count}.

Counting the simplest possible pattern $\pi = 21$ yields the number of inversions, a natural measure of sortedness of a sequence. In the streaming setting, counting 21 already requires $\Omega(n)$ space~\cite{AjtaiEtAl2002}, even when the input sequence is a permutation.

\paragraph{Our results.}

Let us repeat the definition of the two problems we study in this paper (we define the model formally in \cref{sec:prelims:streaming}).
The first problem, called $\pi$-$\PPM_n$, consists of determining whether the pattern $\pi$ is contained in an input permutation $\tau$ on $[n]$ that arrives one value at the time. We treat this as a non-uniform problem, that is, an algorithm may hard-code the input size $n$. This makes our lower bounds stronger and simplifies our algorithms somewhat.

The second problem, $\pi$-$\SeqPPM_n$ is similar, except the input now is a sequence of \emph{at most} $n$ distinct values out of $[n]$.
Our lower bounds still work when the input size $m \le n$ is known to the algorithm, though we do not study all possible combinations of parameters $m, n$.

Note that we can store the whole input using $\fO( n \log n)$ bits, which is therefore a trivial upper bound for the space complexity. It is also not hard to show the following weak lower bound (\cref{sec:ppm:lb}).

\begin{restatable}{theorem}{restateEasyLB}\label{p:lb-easy}
	The space complexity of $\pi$-$\PPM_n$ and $\pi$-$\SeqPPM_n$ for a pattern $\pi$ of length at least two is $\Omega(\log n)$ bits.
\end{restatable}

We now present our main results. For monotone patterns, a simple modification of the standard \emph{longest increasing subsequence} algorithm allows us to solve both problems with $\fO(k \log n)$ bits of space, and we prove that this is tight up to a constant \emph{independent} of $k$.

\begin{restatable}{theorem}{restateMon}\label{p:compl-mon}
	The space complexity of the problems $\pi$-$\PPM_n$ and $\pi$-$\SeqPPM_n$ for each monotone pattern $\pi$ of length $k \ge 2$ is $\Theta(k \log n)$ bits if $n > k^2$.
\end{restatable}

We prove \cref{p:compl-mon} in \cref{sec:ppm:mon}.
For non-monotone patterns, both problems are hard in almost all cases (\cref{sec:ppm:lb}). For $\SeqPPM$, we give a short proof using our machinery.\footnote{Jelínek, Opler, and Valtr~\cite{JelinekOplerEtAl2024PPMMerge} gave another proof based on a preliminary version of this paper.}

\begin{restatable}{theorem}{restateNonMonSeq}\label{p:lb-non-mon-seq}
	The space complexity of $\pi$-$\SeqPPM_n$ for each non-monotone pattern $\pi$ is $\Omega_\pi(n)$ bits.
\end{restatable}

Hardness of $\PPM$ (for almost all patterns) is one of our main new results, and significantly harder to prove.

\begin{restatable}{theorem}{restateNonMon}\label{p:lb-non-mon}
	The space complexity of $\pi$-$\PPM_n$ for each non-monotone pattern $\pi$ of length at least four is $\Omega_\pi(n)$ bits.
\end{restatable}

Finally, we show that $\pi$-$\PPM_n$ can be solved using sublinear space for non-monotone patterns of size three, which implies a separation between $\SeqPPM$ and $\PPM$ (\cref{sec:ppm:size-3-ub}).

\begin{restatable}{theorem}{restateNonMonThree}\label{p:ub-non-mon-3}
	The space complexity of $\pi$-$\PPM_n$ is
	\begin{itemize}
		\item $\fO(\sqrt{n\log{n}})$ bits if $\pi \in \{312,132\}$, and
		\item $\fO(\sqrt{n} \log n)$ bits if $\pi \in \{231,213\}$.
	\end{itemize}
\end{restatable}

Note that all bounds are tight up to a $\fO( \log n)$ factor, except \cref{p:ub-non-mon-3}, where a $\widetilde{\Theta}(\sqrt{n})$ gap remains. Our lower bounds in \cref{p:lb-non-mon-seq,p:lb-non-mon} extend to a \emph{multi-pass} variant: When the algorithm can take $p$ passes over the input, we obtain lower bounds of $\Omega(n/p)$ bits in both cases. Moreover, \cref{p:lb-non-mon-seq,p:lb-non-mon} hold even when randomized or nondeterministic streaming algorithms are allowed (cf.\ Jelínek, Opler, Valtr~\cite{JelinekOplerEtAl2024PPMMerge}).

We conclude the paper with some open questions (\cref{sec:conc}).

\section{Preliminaries}\label{sec:prelims}

We start with some basic notation around sequences and permutations.
Let $\sigma$ be a sequence. We write $|\sigma|$ for the length of the sequence, and $\Val(\sigma)$ for the set of values occurring in $\sigma$. For an integer $k$ with $1 \le k \le |\sigma|$, we let $\sigma(k)$ denote the $k$-the element of $\sigma$. If $\sigma, \tau$ are sequences, then $\sigma \tau$ denotes their concatenation. Sequences are usually written as concatenations of characters $x_1 x_2 \dots$, but also sometimes as tuples $(x_1, x_2, \dots)$ when needed for clarity.

A sequence $\pi$ is a \emph{permutation} of $[n]$ if each element of $[n]$ occurs exactly once.
A permutation \emph{pattern} $\pi$ is \emph{contained} in an integer sequence $\tau$ if $\tau$ has a subsequence $\sigma$ that is \emph{order-isomorphic} to~$\pi$. Formally, $\pi$ and $\sigma$ are order-isomorphic if $|\pi| = k = |\sigma|$ and $\pi(i) < \pi(j) \iff \sigma(i) < \sigma(j)$ for all $i, j \in [k]$. We sometimes call $\sigma$ an \emph{occurrence} of $\pi$ in $\tau$. An occurrence of 12 is called an \emph{increasing pair}, and an occurrence of 21 is called a \emph{decreasing pair}.

\subsection{Streaming problems}\label{sec:prelims:streaming}

A \emph{streaming problem} $(\Sigma, D, P)$ consists of an \emph{alphabet} $\Sigma$ (the set of possible elements appearing in the stream), a finite \emph{domain} $D \subseteq \Sigma^*$ (the set of possible input streams) and a set $P \subseteq D$ of \emph{YES-instances}.

Let $S_n$ denote the set of permutations of $[n]$ and let $S'_n \supset S_n$ denote the set of sequences with pairwise distinct elements of $[n]$. Let
\begin{align*}
	& \text{$\pi$-$\PPM_n$} = ([n], S_n, \{ \tau \in S_n \mid \text{$\tau$ contains $\pi$} \});\text{ and}\\
	& \text{$\pi$-$\SeqPPM_n$} = ([n], S'_n, \{ \tau \in S'_n \mid \text{$\tau$ contains $\pi$} \}).
\end{align*}

We formally model a \emph{streaming algorithm} $\mathbb{A}$ for the problem $(\Sigma, D, P)$ as a state machine $\mathbb{A} = (Q, q_0, F, t)$, where $Q$ is the set of \emph{states}, $q_0 \in Q$ is the \emph{initial state}, $F \subseteq Q$ is the set of \emph{accepting states}, and $t \colon Q \times \Sigma \rightarrow Q$ is the \emph{transition function}. Running the streaming algorithm on an input $\tau = x_1 x_2 \dots x_n \in D$ means computing $q_1 = t(q_0,x_1)$, then $q_2 = t(q_1,x_2)$, and so on.

The algorithm is \emph{correct} if for each input $\tau \in D$, the final state $q^*$ satisfies $q^* \in F$ if and only if $\tau \in P$. Note that the algorithm may behave arbitrarily if $\tau \notin D$.
The \emph{space usage} of $\mathbb{A}$ is $\ceil{\log |Q|}$, i.e., the number of bits required to encode its states.

Observe that this definition allows for arbitrary computation, and is only restricted by space usage \emph{between two input elements}; the transition function may perform arbitrary computation. Our lower bounds in this model thus hold for a wide range of computational models.
On the other hand, the algorithms in \cref{sec:ppm:mon,sec:ppm:size-3-ub} can be implemented in the standard Word-RAM model\footnote{Space usage is still measured in bits, not in words.}, and run in polynomial time.

Most of our lower bounds also hold for algorithms that make multiple passes. Formally, a \emph{$k$-pass streaming algorithm} receives $k$ copies of $\tau$ in sequence (instead of only one), and has to decide $\tau \in P$ only after receiving all of them.

\subsection{Communication complexity}

\newcommand{\CCIndex}{\textsc{Index}}
\newcommand{\CCDisj}{\textsc{Disj}}
\newcommand{\CCUDisj}{\textsc{UDisj}}

To prove our lower bounds, we will use tools from Communication Complexity \cite{Yao1979CC,Lovasz1990,KushilevitzNisan1997, rao2020communication}. We now give some necessary background.
Let $A$ and $B$ be finite sets, and let $Q \subseteq A \times B$. Suppose we have two players, Alice, who knows some $a \in A$, and Bob, who knows some $b \in B$. Alice and Bob each want to decide whether $(a,b) \in Q$.
We are interested in how much information they have to send to each other to determine whether $(a,b) \in Q$. On the other hand, we do not care how much time they need, and assume both may compute arbitrary functions.
The minimal number of bits Alice and Bob have to exchange in the worst case is called the \emph{deterministic communication complexity} of~$Q$, and denoted by $\DCC(Q)$. For more details, we refer to the textbook by Kushilevitz and Nisan \cite{KushilevitzNisan1997}.

The concrete problem we consider is called \emph{disjointness}. Here, Alice and Bob each get a set $S, T \subseteq [n]$, and need to determine whether $S$ and $T$ are disjoint. Formally, let $A = B = 2^{[n]}$, and let $\CCDisj_n = \{ (S, T) \in A \times B \mid S \cap T = \emptyset \}$.

\begin{theorem}[Kushilevitz and Nisan~{\cite[Example 1.21]{KushilevitzNisan1997}}]
	$\DCC(\CCDisj_n) = n+1$.
\end{theorem}

The randomized (and even nondeterministic) communication complexity of $\CCDisj_n$ is also known to be at least $n$~\cite{KushilevitzNisan1997}. This is the reason that most of our lower bounds extend to these models; we omit further details.

Let $(\Sigma, D, P)$ be a streaming problem and $Q\subseteq A \times B$ be a two-party communication complexity problem. We say $P$ is \emph{$k$-$Q$-hard} if there exists functions $f_1, f_2, \dots, f_k \colon A \rightarrow \Sigma^*$ and $g_1, g_2, \dots, g_k \colon B \rightarrow \Sigma^*$ such that, for each input $(a,b) \in A \times B$, the sequence \[\sigma = f_1(a) g_1(b) f_2(a) g(b) \dots f_k(a) g_k(b)\] is contained in $D$ and we have $\sigma \in P$ if and only if $(a,b) \in Q$.

\begin{lemma}\label{p:cc-lb}
	If a streaming problem is $k$-$Q$-hard, then every algorithm solving it with $p \in \N_+$ passes requires at least $\DCC(Q) / (2kp)$ bits of space.
\end{lemma}
\begin{proof}
	Let $\mathbb{A}$ be a correct algorithm with $k$ passes. We construct a communication complexity protocol for $Q$ as follows. Recall that Alice receives $a \in A$ and Bob receives $b \in B$. Alice initializes~$\mathbb{A}$, runs it on $f_1(a)$, and passes the internal state of the algorithm to Bob. Bob continues running the algorithm on $g_1(b)$, and passes the resulting state to Alice. Alice and Bob continue this until they have cycled through the functions $f_1, f_2, \dots, f_k$ and $g_1, g_2, \dots, g_k$ $p$ times. After that, $\mathbb{A}$ outputs $(a,b) \in Q$ by definition, so both Alice and Bob know the result.
	
	Suppose $\mathbb{A}$ requires $s$ bits of space. Then at most $s$ bits are sent in every step. Since there are $2 k \cdot p$ steps, the total number of exchanged bits is $2kp s$. By definition of the deterministic communication complexity, we have $2kps \ge \DCC(Q)$, which implies the statement.
\end{proof}

\subsection{Complement and reverse}

We now show some basic facts about how our two problems behave under taking the complement or reverse of the pattern.
If $\pi$ is a sequence of distinct values in $[n]$, the \emph{complement} $\pi^\complement$ of $\pi$ is obtained by replacing each value $i$ by $n+1-i$.

\begin{lemma}\label{p:complement}
	Let $\pi$ be a permutation. Then, for each $n \in \N_+$, the space complexities of $\pi$-$\PPM_n$ and $\pi^\complement$-$\PPM_n$ are equal, and the space complexities of $\pi$-$\SeqPPM_n$ and $\pi^\complement$-$\SeqPPM_n$ are equal.%
\end{lemma}
\begin{proof}
	We can easily compute the complement of the input sequence by complementing each individual value. Thus, an algorithm that finds $\pi$ in an input $\tau$ can be made to find $\pi$ in $\tau^\complement$, which is equivalent to finding $\pi^\complement$ in $\tau$.
\end{proof}

Note that equality in \cref{p:complement} holds in our model (\cref{sec:prelims:streaming}). If we measure space usage while processing a single input element, we may incur an $\fO(\log n)$ overhead needed to compute the complement, depending on the exact model. However, an additive $\fO(\log n)$ term does not matter for any of our results.

We do not know whether an analog of \cref{p:complement} holds for taking the \emph{reverse} of a permutation. However, our main lower bound technique, using communication complexity, is easily seen to be invariant under taking the reverse:

\begin{observation}\label{p:cc-hard-reverse}
	Let $\pi$ be a permutation and let $\rho$ be the reverse of $\pi$. If $\pi$-$\PPM_n$ is $k$-$Q$-hard for some communication complexity problem $Q$, then $\rho$-$\PPM_n$ is also $k$-$Q$-hard.
\end{observation}

\section{Monotone patterns}\label{sec:ppm:mon}

In this section, we discuss the complexity of $\pi$-$\PPM_n$ and $\pi$-$\SeqPPM_n$ when $\pi$ is an \emph{increasing} pattern $1 2 \dots k$. The same results hold for \emph{decreasing} patterns $k (k-1) \dots 1$, by \cref{p:complement}.

To prove the upper bound, we use the well-known \emph{longest increasing subsequence} algorithm~\cite{Schensted1961,Knuth1973,Fredman1975} with an early stopping condition. Suppose we want to detect the increasing pattern of length $k$. The algorithm works as follows: For each $i \in [k]$, maintain the smallest value $x_i$ seen so far that is the last value in an increasing subsequence of length exactly $i$. Let $x_i = \bot$ as long as no such element exists. Updating each $x_i$ when a new value $y$ arrives is simple: For each $i \in [k]$, if $x_i < y$ and $y < x_{i+1}$ or $x_{i+1} = \bot$, set $x_{i+1} \gets y$. Accept as soon as $x_k \neq \bot$.
Since the algorithm stores at most $k$ input values at any time, we have:

\begin{lemma}
	If $\pi$ is a monotone permutation of length $k$, then $\pi$-$\SeqPPM_n$ (and thus $\pi$-$\PPM_n$) can be solved with $k \cdot \ceil{\log n}$ bits of space.
\end{lemma}

We now show a tight lower bound for all $k \ge 4$ that holds even if the input is a permutation. %

\begin{lemma}\label{p:lb-mon-4}
	If $\pi$ is the increasing or decreasing permutation of length $k \ge 4$, then solving the $\pi$-$\PPM_n$ problem requires $\Omega( k \log n)$ bits of space if $n > k^2$.
\end{lemma}
\begin{proof}	
	We give a direct communication complexity proof. Assume $n$ is even and $n \ge k^2$. Let Alice have the first half of the input permutation, which will consist of the odd integers in $[n]$, and let Bob have the latter half, containing the even integers of $[n]$. We will construct a set of $2^{\Omega(k \log n)}$ possible inputs for Alice such that Bob must be able to distinguish between any two of them. This means that Alice has to send $\Omega( k \log n)$ bits of information, implying the stated lower bound.
	
	Let $\rho = r_1 r_2 \dots r_{k-2} \in [n]^{k-2}$ be an increasing sequence of odd integers with $r_1 = 1$. Observe that there are
	\[ \binom{n/2 - 1}{k-3} \ge \left(\frac{n}{k}\right)^{\Omega(k)} \ge (\sqrt{n})^{\Omega(k)} = 2^{\Omega(k \log n)} \]
	 such sequences.
	We claim that we can construct a permutation $\alpha$ for Alice so that $\rho$ describes the state of the algorithm given at the start of the section; that is, we have that $r_i$ is the lowest value that is the last element of an increasing sequence of length $i$. Indeed, let $\alpha$ start with a decreasing sequence of all odd integers \emph{not} contained in $\rho$, and end with $\rho$.
	
	\begin{figure}
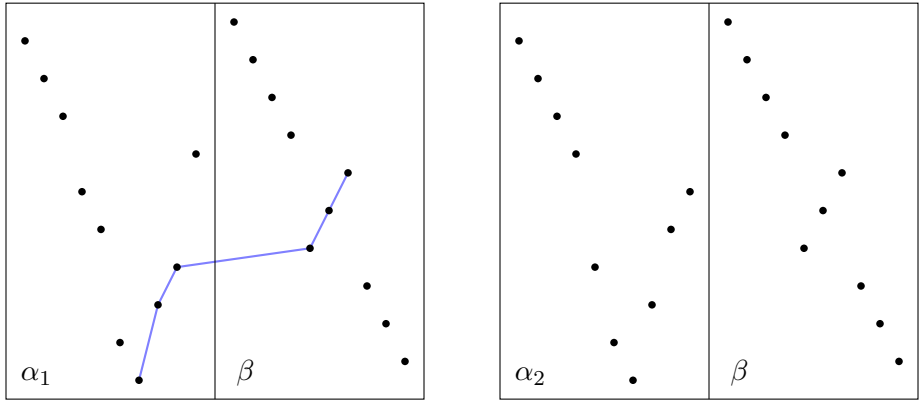

		\centering
		\begin{tikzpicture}[
			scale=0.25,
			point/.style={fill, circle, inner sep=1pt}
			]
			\draw[blue!50!white, thick] (7,1) -- (8,5) -- (9,7) -- (16,8) -- (17,10) -- (18,12);
			\input{figs/red_PPM_mon.tex}
			\node[anchor=base,xshift=1.5mm] at (1,1) {$\alpha_1$};
			\node[anchor=base,xshift=1.5mm] at (12,1) {$\beta$};
			\begin{scope}[shift={(26,0)}]
				\input{figs/red_PPM_mon2.tex}
				\node[anchor=base,xshift=1.5mm] at (1,1) {$\alpha_2$};
				\node[anchor=base,xshift=1.5mm] at (12,1) {$\beta$};
			\end{scope}
		\end{tikzpicture}
		\caption{Constructions $\alpha_1 \beta$ (left) and $\alpha_2 \beta$ (right) for \cref{p:lb-mon-4} with $k = 6$, $n = 20$, $\rho = (1,5,7,13)$, and $\sigma = (1,5,9,11)$. Observe that $\alpha_1 \beta$ contains an increasing subsequence of length 6 (marked in blue), but $\alpha_2 \beta$ does not.}\label{fig:red:mon-4}
	\end{figure}
	
	Now let $\alpha_1, \alpha_2$ be constructed in this way from distinct integer sequences $\rho = r_1 r_2 \dots r_{k-1}$ and $\sigma = s_1 s_2 \dots s_{k-1}$. We show that there exists a permutation $\beta$ of the even values in $[n]$ so that out of $\alpha_1 \beta$ and $\alpha_2 \beta$, exactly one contains $\pi$. That means that Bob, when given $\beta$, must be able to distinguish between $\alpha_1$ and $\alpha_2$, as desired. (See \cref{fig:red:mon-4}.)
	
	Let $i \in [k-1]$ be minimal such that $r_i \neq s_i$, and assume $r_i < s_i$ without loss of generality. Let now $\beta = \beta_1 \beta_2 \beta_3$ be the following sequence. Start with the decreasing subsequence $\beta_1 = (n, n-2, \dots, r_i+2(k-i)+1)$, continue with increasing $\beta_2 = (r_i+1, r_i+3, \dots, r_i+2(k-i)-1)$ and finish with $\beta_3 = (r_i-1, r_i-3, \dots, 2)$. Observe that $\alpha_1 \beta$ contains an increasing subsequence of length $k$. Indeed, recall that $r_i$ is the last element of an increasing subsequence of length $i$, which can be concatenated with $\beta_2$ to obtain one of length $k$.
	
	We finish the proof by arguing that $\alpha_2 \beta$ does not contain a increasing subsequence of length $k$. Suppose for the sake of contradiction that such a subsequence $\gamma = \gamma_1 \gamma_2$ exists, with $\gamma_1$ being part of $\alpha_2$ and $\gamma_2$ being part of $\beta$. Write $j = |\gamma_1|$, so $k-j = |\gamma_2|$. The longest increasing subsequence in $\beta$ has length $k-i$, implying $j \ge i$. Further, since the longest increasing subsequence in $\alpha_2$ has length $k-2$, we have $k-j \ge 2$.
	
	Note that $\gamma_1$ ends with a value at least $s_j \ge s_i + 2(j-i) > r_i + 1 + 2(j-i)$, and $\gamma_2$ starts with a value at most $r_i + 2(j-i+1)-1 < s_j$, which means $\gamma_1 \gamma_2$ is not increasing, a contradiction.
\end{proof}

For smaller patterns, we can still show a $\Omega(\log n) = \Omega(k \log n)$ bound. Note that this is easy for $\pi$-$\SeqPPM_n$, even if $\pi= 12$ and the input sequence consists of only two values $x, y$: By a simple adversarial argument, a correct algorithm must be able to distinguish between any two odd values for $x$. For $\pi$-$\PPM$, the proof is a little more involved.

\begin{lemma}\label{p:lb-mon2}
	Solving the $\pi$-$\PPM_n$ problem for monotone $\pi$ of length at least 2 requires $\Omega(\log n)$ space.
\end{lemma}
\begin{proof}
	Set $\pi = 21$ for now. Let $\Algo{A}$ be a streaming algorithm $\Algo{A}$ for $\PPM_n$ with $m$ states. In the following, we show that $m \ge \tfrac{n}{2} + 1$, implying that $\log n$ bits of space are necessary.
	
	Let $q_0, q_1, \dots, q_n$ be the sequence of states $\Algo{A}$ takes when the input is the increasing permutation $12\dots n$. Observe that states may repeat. However, we claim that pairs of consecutive states must be pairwise distinct. Indeed, suppose we have $q_i = q_j$ and $q_{i+1} = q_{j+1}$ for two distinct $i, j \in \{0,1,\dots, n\}$. Then swapping $i$ and $j$ in the input will not change the behavior of~$\Algo{A}$, even though the sequence is no longer monotone.
	
	Our claim implies that the state sequence is a so-called \emph{Davenport-Schinzel sequence} of order~2~\cite{AgarwalSharir2000DavenportSchinzel}. It is well known that if such a sequence consists of $m$ distinct symbols (corresponding to \emph{states} here), then its length is at most $2m-1$. Thus, we have $n+1 \le 2m-1$, implying that $m \ge \tfrac n 2 + 1$, as desired.
	
	The case $\pi = 321$ reduces to 21 by simply adding the value $n+1$ at the start. The cases $\pi \in \{12, 123\}$ are symmetric.
\end{proof}

Overall, since $\pi$-$\PPM$ cannot be harder than $\pi$-$\SeqPPM$, we have:
\restateMon*

\section{Lower bounds for non-monotone patterns}\label{sec:ppm:lb}

In this section, we prove \cref{p:lb-easy,p:lb-non-mon,p:lb-non-mon-seq}.
We first prove a useful reduction, essentially showing that \emph{appending} or \emph{prepending} a value to the pattern cannot only make our pattern matching problems (much) easier.

\begin{lemma}\label{p:red}
	Let $\pi$ be a pattern of length $k \ge 1$ and let $\pi'$ be a pattern of length $k+1$ such that $\pi' = \sigma x$ or $\pi' = x \sigma$ for some $x \in [n]$ and a sequence $\sigma$ that is order-isomorphic to $\pi$.
	
	Fix a number $p \in \N_+$ of passes.
	Suppose that $\pi'$-$\textsc{SequencePPM}_n$ ($\pi'$-$\textsc{PPM}_n$) can be solved with $f(n)$ bits of space and $p$ passes for every $n$. Then, $\pi$-$\textsc{SequencePPM}_n$ ($\pi$-$\textsc{PPM}_n$) can be solved with $f(2n) + \fO(\log n)$ bits of space and $p$ passes for every~$n$.
\end{lemma}
\begin{proof}
	Let $\pi' = \sigma x$; the other case is similar.
	Take an algorithm $\Algo{A}$ that solves $\pi'$-$\textsc{PPM}_{2n}$ with $f(2n)$ bits of space. We now give an algorithm for $\pi$-$\textsc{PPM}_n$. Without loss of generality, assume that $\pi'(k-1) > \pi'(k)$.
	
	Let $\tau$ be the input. Compute $\tau'$ by first taking every value in $\tau$, in order, and doubling it; and then append the sequence $1, 3, 5, \dots, 2n-1$. Observe that $\tau'$ is a permutation on $[2n]$ and can be constructed one value at the time as $\tau$ arrives, using $\fO(\log n)$ space. The algorithm for $\pi$-$\textsc{PPM}_n$ consists of feeding $\tau'$ into $\Algo{A}$. The algorithm clearly requires $f(2n) + \fO(\log n)$ space. We now show correctness.
	
	If $\tau$ contains $\pi$, then clearly $\tau'$ also contains $\pi'$. On the other hand, if $\tau'$ contains $\pi'$, then either the occurrence of $\pi'$ is within the first $n$ elements of $\tau'$ (which implies $\tau$ contains $\pi'$, and thus $\pi)$, or, since $\pi'(k-1) > \pi'(k)$, only the last element of $\pi'$ is mapped to the second half of $\tau'$ (which implies that $\tau$ contains $\pi$). Thus, our algorithm is correct.
	
	The proof for $\textsc{SequencePPM}$ is essentially the same.
\end{proof}

It seems likely that the space complexity of $\pi'$-$\textsc{SequencePPM}_n$ and $\pi'$-$\textsc{PPM}_n$ is, in fact, \emph{monotone} under taking sub-patterns (say, for large enough $n$); however, this seems more difficult to prove and is not needed to show our bounds.

Combining \cref{p:red} with \cref{p:lb-mon2} already yields a $\Omega(\log n)$ bound for every pattern $\pi$, thus proving \cref{p:lb-easy}.
We now move on to the lower bound for $\SeqPPM{}$.

\restateNonMonSeq*

By \cref{p:red}, it suffices to show a lower bound for non-monotone patterns of length 3. We show the stronger lower bound in the multi-pass model, as mentioned in the introduction.

\begin{figure}
	\centering
	\begin{tikzpicture}[
			scale=0.25,
			point/.style={fill, circle, inner sep=1pt}
		]
		\input{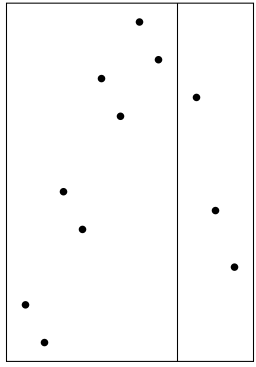}
	\end{tikzpicture}
	\caption{Construction of $\tau = f_1(S) g_1(S)$ for \cref{p:lb-non-mon-seq-3} with $S = \{1,3,5,6\}$ and $T = \{2,3,5\}$.}\label{fig:red:seq-312}
\end{figure}

\begin{lemma}\label{p:lb-non-mon-seq-3}
	The space complexity of $\pi$-$\textsc{SequencePPM}_n$ with $p$ passes is $\Omega(n/p)$ bits for each non-monotone pattern $\pi$ of length three.
\end{lemma}
\begin{proof}
	We show that 312-$\SeqPPM_{3n}$ is 1-$\CCDisj_n$ hard, which suffices to prove our claim by \cref{p:cc-lb,p:complement,p:cc-hard-reverse}.
	Let $(S,T)$ be an input for $\textsc{Disj}_n$. Construct the sequence $f_1(S)$ as follows (see \cref{fig:red:seq-312}). For each $i \in S$, in \emph{ascending} order, add the values $3i, 3i-2$. Further, construct the sequence $g_1(T)$ by adding $3i-1$ for each $i \in T$, in \emph{descending} order. We now show that $\tau = f_1(S) g_1(T)$ contains 312 if and only if there is some $i \in S \cap T$.
	
	First suppose that $i \in S \cap T$; then $\tau$ contains the subsequence $3i, 3i-2, 3i-1$, which is an occurrence of 312. Second, suppose $\tau$ contains an occurrence $(a,b,c)$ of 312. Since $g_1(T)$ is descending, we cannot have $b, c \in g_1(T)$, implying $a, b \in f_1(S)$. Since $a > b$, and all descending pairs in $f_1(S)$ are of the form $(3i, 3i-2)$, we have $a = 3i$ and $b = 3i-2$ for some $i \in [n]$, implying $i \in S$. Now we must have $c = 3i-1$, which implies that $i \in T$, and we are done.
\end{proof}

We now formally prove the following, which includes \cref{p:lb-non-mon-seq} as the special case $p = 1$.

\begin{theorem}\label{p:lb-non-mon-seq-passes}
	The space complexity of $\pi$-$\SeqPPM_n$ with $p$ passes for each non-monotone pattern $\pi$ is $\Omega_\pi(n/p - \log n)$ bits.
\end{theorem}
\begin{proof}
	Observe that for each non-monotone pattern $\pi$ of length at least four, we can remove either the first or last element and the pattern stays non-monotone. Thus, every pattern $\pi$ of length $k \ge 3$ can be reduced to a non-monotone pattern $\rho$ of length 3 in $k-3$ such steps. By \cref{p:red}, if $\pi$-$\textsc{SequencePPM}_n$ can be solved with $f(n)$ space and $p$ passes, then $\rho$-$\SeqPPM_n$ can be solved with $f(2^{k-3}n) + \fO_k(\log n)$ space and $p$ passes. Our claim thus follows from \cref{p:lb-non-mon-seq-3}.
\end{proof}

\subsection{\texorpdfstring{$\CCDisj$}{Disj}-hardness of PPM}\label{sec:disj-hardness}

In this section, we show that for each non-monotone pattern $\pi$ of size four, $\pi$-$\PPM_m$ is 2-$\CCDisj_n$-hard for some $m \in \fO(n)$. By \cref{p:cc-lb}, this implies the $\Omega(n/p)$ lower bound for these patterns (with $p$ passes), and, as in the last section, we obtain the following, which includes \cref{p:lb-non-mon} as a special case:
\begin{theorem}
	The space complexity of $\pi$-$\PPM_n$ with $p$ passes for each non-monotone pattern $\pi$ of length at least four is $\Omega_\pi(n/p)$ bits.
\end{theorem}

We show $\CCDisj_n$-hardness for the concrete patterns 4231, 4213, 4132, 4123, 4312, 3142, and 2143.
By \cref{p:complement,p:cc-hard-reverse}, this covers all other non-monotone patterns of length four as well.

\begin{lemma} \label{prop:OLPPM_lower_bound_4abc}
	Let $\pi \in \{4231, 4213, 4132, 4123\}$. Then $\pi$-$\PPM_m$ is 1-$\CCDisj_n$-hard for each $n$ and $m = 4n$.
\end{lemma}
\begin{proof}
	\Cref{fig:streaming:red_OLSPPM_4231_4132} shows examples of our construction.
	Let $S, T \in [n]$. We define $f_1(S) = (a_1, a_2, \dots, a_{2n})$, where for $i \in [n]$,
	\begin{align*}
		(a_{2i-1, 2i}) = \begin{cases}
			(4(i-1) + \pi(1), 4(i-1) + \pi(2)), & \text{ if } i \in S, \\
			(4(i-1) + \pi(2), 4(i-1) + \pi(1)), & \text{ otherwise.}
		\end{cases}
	\end{align*}
	
	For $i \in [n]$, let $d(i) = i$ if $\pi = 4231$, and let $d(i) = n+1-i$ otherwise. Note that always $d(d(i)) = i$. We define $g_1(T) = (b_1, b_2, \dots, b_n)$, where for $i \in [n]$,
	\begin{align*}
		(b_{2i-1}, b_{2i}) = \begin{cases}
			(4(d(i)-1) + \pi(3), 4(d(i)-1) + \pi(4)), & \text{ if } d(i) \in T, \\
			(4(d(i)-1) + \pi(4), 4(d(i)-1) + \pi(3)), & \text{ otherwise.}
		\end{cases}
	\end{align*}
	
	Observe that the sequence $\tau = f_1(S) g_1(T)$ contains each value in $[4n]$ exactly once, i.e. $\tau$ is a $4n$-permutation.
	We need to prove that $\tau$ contains $\pi$ if and only if $S \cap T \neq \emptyset$. Suppose there is some $i \in S \cap T$. As $d(d(i)) = i \in T$, the sequence \[(a_{2i-1}, a_{2i},  b_{2d(i)-1}, b_{2d(i)}) = (4(i-1) + \pi(1), (4(i-1) + \pi(2), (4(i-1) + \pi(3), (4(i-1) + \pi(4))\] is a subsequence of $\tau$, implying that $\tau$ contains $\pi$.
	
	\begin{figure}[tbp]
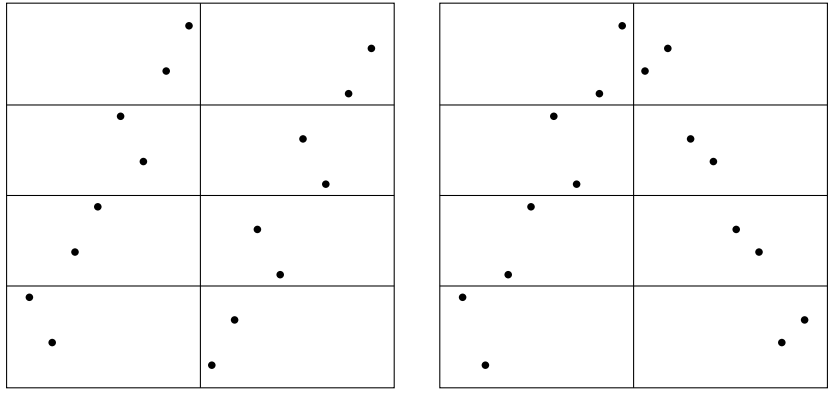

		\centering
		\begin{tikzpicture}[
			scale=0.3,
			point/.style={fill, circle, inner sep=1pt}
			]
			\input{figs/red_OLPPM_4231}
			\begin{scope}[shift={(19,0)}]
				\input{figs/red_OLPPM_4132}
			\end{scope}
		\end{tikzpicture}
		\caption{Construction of $\tau = f_1(S) g_1(S)$ for $\pi = 4231$ (left) and for $\pi = 4132$ (right) with $n = 4$, $S = \{1,3\}$ and $T = \{2,3\}$}
		\label{fig:streaming:red_OLSPPM_4231_4132}
	\end{figure}
	
	Now suppose $\tau$ contains $\pi$, so there are indices $i_1 < i_2 < i_3 < i_4$ such that the subsequence $\tau(i_1) \tau(i_2) \tau(i_3) \tau(i_4)$ is order-isomorphic to $\pi$.
	We need to show that $S \cap T \neq \emptyset$.
	
	We first show that $i_3 > 2n$, i.e., $\tau(i_3)$ appears in $g_1(T)$. Suppose that is not the case. Then, by definition of $f_1(S)$, there is at most one index $i < i_3$ such that $\tau(i) > \tau(i_3)$, namely $i = i_3 - 1$. As $\pi(1) > \pi(3)$, this implies $i_1 = i-1$, so $i-1 = i_1 < i_2 < i_3 = i$, a contradiction.
	
	Second, we show that $i_2 \le 2n$. Assume this is not the case and $\pi = 4231$. As $d(i) = i$, there is at most one index $i > i_2$ with $\tau(i) < \tau(i_2)$, namely $i = i_2 + 1$. As above, this is a contradiction. In the case that $\pi \neq 4231$, we have $\pi(2) < \pi(4)$ and $d(i)$ is decreasing, which similarly implies a contradiction.
	
	We now have $i_1 < i_2 \le 2n < i_3 < i_4$. As $\pi(1) > \pi(2)$, we know that $\tau(i_1) > \tau(i_2)$. This means $i_1 = 2i-1$ and $i_2 = 2i$ for some $i \in S$, as $f_1(S)$ contains no other decreasing pairs. Consequently, we have $\tau(i_j) = 4(i-1) + \pi(j)$ for $j \in [1,2]$.
	
	To conclude the proof, we show that $\tau(i_j) = 4(i-1) + \pi(j)$ for $j \in [3,4]$, which implies $i \in T$ and thus $i \in S \cap T$. We distinguish three cases.
	\begin{itemize}
		\item Let $\pi = 4231$. We have $\pi(2) < \pi(3) < \pi(1)$, thus
		\begin{align*}
			4(i-1) + 2 = \tau(i_2) < \tau(i_3) < \tau(i_1) = 4(i-1) + 4,
		\end{align*}
		implying that indeed $\tau(i_3) = 4(i-1) + 3$. Now $\tau(i_3), \tau(i_4)$ must be a decreasing pair, which leaves only $4(i-1) + 1$ as a possible value for $\tau(i_4)$, since $d(i)$ is increasing.
		
		\item Let $\pi = 4213$. By a similar argument as above, $\pi(2) < \pi(4) < \pi(1)$ implies that $\tau(i_4) = 4(i-1) + 3$. Now $\tau(i_3), \tau(i_4)$ must be an \emph{increasing} pair, and $d(i)$ is decreasing, so $\tau(i_3) = 4(i-1) + 1$.
		
		\item Finally, let $\pi \in \{4132, 4123\}$. Now $\pi(2) < \pi(3) < \pi(1)$ and $\pi(2) < \pi(4) < \pi(1)$, so $\tau(i_3), \tau(i_4) \in \{\tau(i_4) = 4(i-1) + \pi(3), \tau(i_4) = 4(i-1) + \pi(4)\} = \{b_{2i-1}, b_{2i}\}$. Since $i_3 < i_4$, we have $i_3 =  2i-1$ and $i_4 = 2i$, implying $i \in T$.\qedhere
	\end{itemize}
\end{proof}

\begin{figure}[tbp]
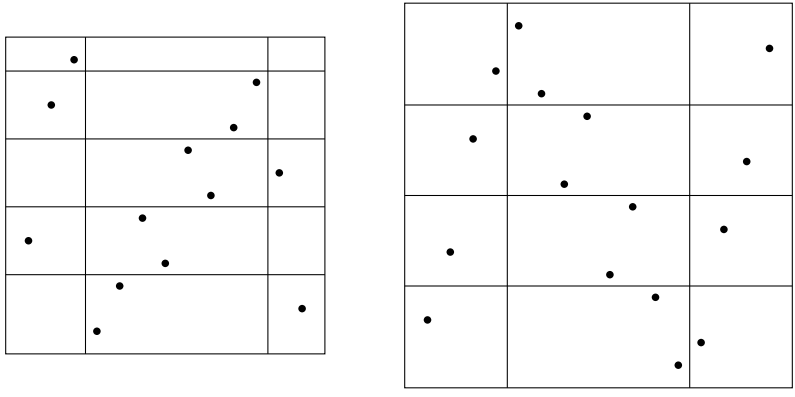

	\centering
	\begin{tikzpicture}[
		scale=0.3,
		point/.style={fill, circle, inner sep=1pt}
		]
		\begin{scope}[shift={(1.5, -17.5)}]
			\input{figs/red_OLPPM_4312}
		\end{scope}
		\begin{scope}[shift={(19, -19)}]
			\input{figs/red_OLPPM_3142}
		\end{scope}
	\end{tikzpicture}
	\caption{Construction of $\tau = f_1(S) g_1(S) f_2(T)$ for $\pi = 4312$ (left) and for $\pi = 3142$ (right) with $n = 4$, $S = \{1,3\}$ and $T = \{2,3\}$.}
	\label{fig:streaming:red_OLSPPM_4312_3142}
\end{figure}

\begin{lemma} \label{prop:OLPPM_lower_bound_4213}
	$4312$-$\PPM_m$ is $2$-$\CCDisj_n$-hard for all $n$ and $m = 3n+1$.
\end{lemma}
\begin{proof}
	\Cref{fig:streaming:red_OLSPPM_4312_3142} (left) shows an example of the construction.
	Let $S, T \subseteq [n]$. Let $f_1(S)$ consist of the values $3(i-1) + 2$ for $i \notin S$ in ascending order, and the value $3n+1$ at the end. Let $g_1(T) = b_1 b_2 \dots, b_{2n}$, where
	\begin{align*}
		(b_{2i-1}, b_{2i}) = \begin{cases}
			(3(i-1) + 3, 3(i-1) + 1), & \text{ if } i \in T, \\
			(3(i-1) + 1, 3(i-1) + 3), & \text{ otherwise}.
		\end{cases}
	\end{align*}
	
	Finally, let $f_2(S)$ contain the values $3(i-1) + 2$ for $i \in S$, in \emph{descending} order, and let $g_2(T)$ be empty. The sequence $\tau = f_1(S) g_1(T) f_2(S)$ is clearly a $(3n+1)$-permutation. We claim that $\tau$ contains 4231 if and only if $S \cap T \neq \emptyset$.
	
	Suppose $i \in S \cap T$. Then $b_{2i-1} = 3(i-1) + 3$, $b_{2i} = 3(i-1) + 1$ and the value $3(i-1)+2$ occurs in $f_2(S)$. The value $3n+1$ always occurs in $f_1(S)$. These four values together form an occurrence of $4312$ in $\tau$.
	
	Now suppose 4312 occurs in $\tau$, so there are indices $i_1 < i_2 < i_3 < i_4$ such that the subsequence $\tau(i_1) \tau(i_2) \tau(i_3) \tau(i_4)$ is order-isomorphic to 4312. Let $k_1 = |f_1(S)|$ and $k_2 = |g_1(T)|$. First, as $f_1(S)$ is increasing and  $f_2(S)$ is decreasing, we know $k_1 < i_2$ and $i_3 \le k_1+k_2$. Second, we claim that $i_1 \le k_1$. Assume this is not the case. Then $k_1 < i_1 < i_2 < i_3 < k_2$, so the decreasing sequence $\tau(i_1) \tau(i_2) \tau(i_3)$ is a subsequence of~$g_1(T)$. However, $g_1(S)$ does not contain 321, a contradiction. Third, $k_1+k_2 < i_4$, as $g_1(T)$ does not contain the pattern 312. In conclusion, we have $i_1 \le k_1 < i_2 < i_3 \le k_1+k_2 < i_4$.
	
	By definition of $g_1(T)$ and using $\tau(i_2) > \tau(i_3)$, we know that there is some $i \in T$ such that $\tau(i_2) = 3(i-1) + 3$ and $\tau(i_3) = 3(i-1) + 1$. Thus $\tau(i_4) = 3(i-1) + 2$. This value occurring in $f_2(S)$ implies that $i \in S$.
\end{proof}

\begin{lemma} \label{prop:OLPPM_lower_bound_2143_2413}
	$3142$-$\PPM_m$ and $2143$-$\PPM_m$ are $2$-$\CCDisj_n$-hard for all $n$ and $m = 4n$.
\end{lemma}
\begin{proof}
	\Cref{fig:streaming:red_OLSPPM_4312_3142} (right) shows an example of the construction.
	Let $S, T \subseteq [n]$. Let $f_1(S) = a_1 a_2 \dots a_n$ and $f_2(S) = c_1 c_2 \dots c_n$, where for $i \in [n]$,
	\begin{align*}
		& a_i = \begin{cases}
			4(i-1) + \pi(1), & \text{ if } i \in S, \\
			4(i-1) + \pi(4), & \text{ otherwise.}
		\end{cases} \\
		& c_i = \begin{cases}
			4(i-1) + \pi(4), & \text{ if } i \in S, \\
			4(i-1) + \pi(1), & \text{ otherwise.}
		\end{cases}
	\end{align*}
	
	Let $g_1(S) = b_1 b_2 \dots b_{2n}$ where for $i \in [n]$,
	\begin{align*}
		(b_{2i-1}, b_{2i}) = \begin{cases}
			(4(n-i) + 1, 4(n-i) + 4), & \text{ if } i \in T, \\
			(4(n-i) + 4, 4(n-i) + 1), & \text{ otherwise}.
		\end{cases}
	\end{align*}
	
	Let $g_2(T)$ be empty. The sequence $\tau = f_1(S) g_1(T) f_2(S)$ is clearly a $4n$-permutation.
	We show that $\tau$ contains $\pi$ if and only if $S \cap T \neq \emptyset$.
	
	First, suppose there is some $i \in S \cap T$. Then $(a_{i}, b_{2i-1}, b_{2i}, c_{i})$ is an occurrence of $\pi$ in $\tau$.
	
	Now suppose $\tau$ contains $\pi$, so there are indices $i_1 < i_2 < i_3 < i_4$ such that the subsequence $\tau(i_1) \tau(i_2) \tau(i_3) \tau(i_4)$ is order-isomorphic to $\pi$. Both $f_1(S)$ and $f_2(S)$ are increasing. As $\pi(1) > \pi(2)$ and $\pi(3) > \pi(4)$, we have $n < i_2 \le i_3 \le 3n$. Moreover, since $\pi(2) < \pi(3)$, we know that $(i_2, i_3) = (n+2i-1, n+2i)$ for some $i \in T$, as there are no other increasing pairs in $g_1(T)$. This means $\tau(i_2) = 4(i-1) + 1$ and $\tau(i_3) = 4(i-1) + 4$.	Now, as $\pi(1)$ and $\pi(4)$ are between $\pi(2) = 1$ ad $\pi(3) = 4$, we have $\tau(i_1) = 4(i-1) + \pi(1)$ and $\tau(i_4) = 4(i-1) + \pi(4)$, which is only possible if $i \in S$.
\end{proof}

\section{Algorithms for size-3 patterns}\label{sec:ppm:size-3-ub}

In this section, we show that the complexity of $\pi$-$\PPM_n$ is $\widetilde{\fO}(\sqrt{n})$ if $\pi \in \{312, 132, 231, 213\}$.
\restateNonMonThree*

We give an algorithm for 312 (which also holds for 132 by \cref{p:complement}) in \cref{sec:ppm:algo-312}, and a different algorithm for 231 (which also holds for 213) in \cref{sec:ppm:algo-213}.
Besides the $\fO(\sqrt{\log n})$ gap in space complexity, there is another interesting difference. The algorithm for 312 actually finds and reports the three values of the occurrence of 312, if any. On the other hand, the algorithm for 231 may not be able to report all values, and instead only reports whether an occurrence exists or not.

\subsection{An algorithm for 312-PPM}\label{sec:ppm:algo-312}

In this section, we give a $\fO( \sqrt{n \log n} )$-space algorithm for 312-$\PPM_n$.
Let $k = \lfloor \sqrt{n \log n} \rfloor$. Throughout the algorithm, we maintain the following data.
\begin{itemize}
	\itemsep0pt
	\item The highest value $h \in [n]$ encountered so far.
	\item A set $D$ of value pairs $(a,b)$ with $a>b$, each representing a decreasing pair encountered before.
	\item A set $A$ with all values in $\{h-k+1, h-k+2, \dots, h\}$ encountered so far.%
\end{itemize}

At the start, we read the first value $v$ and set $h = v$, $D \gets \emptyset$, and $A = \{v\}$.
\Cref{alg:312} shows pseudocode for one step of the algorithm, were we read a value $v$ and either update our data structures or report an occurrence of the pattern 312.

\begin{algorithm}
	\caption{One step of the 312-detection algorithm}\label{alg:312}
	\begin{algorithmic}[1]
		\Procedure{Step}{$v$}
		\If{$\exists (a,b) \in D: a > v > b$}\label{line:check-D}
		\State Report occurrence $(a,b,v)$ and stop\label{line:occ-D}
		\EndIf
		\If{$v > h$}
		\State Delete from $A$ all values smaller than $v-k$\label{line:delete-A}
		\State Insert $v$ into $A$.\label{line:insert-A-1}
		\State $h \gets v$\label{line:update-h}
		\ElsIf{$h-k < v \le h$}\label{line:if-into-A}
		\If{$\exists a, c: a \in A, c \notin A, a > c > v$}\label{line:if-report-A}
		\State Report occurrence $(a,v,c)$ and stop\label{line:occ-A}
		\Else
		\State Insert $v$ into $A$\label{line:insert-A-2}
		\EndIf
		\ElsIf{$v \le h-k$}
		\State Delete all $(a,b)$ with $b > v$ from $D$\label{line:delete-D}
		\State Insert $(h,v)$ into $D$\label{line:insert-D}
		\EndIf
		\EndProcedure
	\end{algorithmic}
\end{algorithm}

We now show that the following invariants hold until the algorithm stops. Let $\tau$ be the input permutation, and let $\sigma$ be the prefix of $\tau$ read so far.
\begin{enumerate}[(i)]
	\itemsep0pt
	\item $h$ is the highest value in $\sigma$.\label{inv:h}
	\item $A = \Val(\sigma) \setminus [h-k]$.\label{inv:A-exact}
	\item For each decreasing pair $(a,b)$ in $\sigma$ with $a > b > h-k$, there is no value $c$ in $\tau$ such that $a,b,c$ is an occurrence of 312 in $\tau$.\label{inv:A-no-occ}
	\item For each $(a,b) \in D$, we have $a-b \ge k$, and $(a,b)$ is a decreasing pair in $\sigma$.\label{inv:D-corr-high}
	\item For each $(a,b), (a',b') \in D$, we have $[b,a] \cap [b',a'] = \emptyset$.\label{inv:D-disjoint}
	\item For each occurrence $a,b,c$ of 312 in $\tau$, where $a$ and $b$ occur in $\sigma$, there is a pair $(a',b') \in D$ such that $a' \ge a > b \ge b'$.\label{inv:D-max}
\end{enumerate}

\begin{lemma}
	The above invariants are true at the start and after every step of the algorithm, except after a step where the algorithm reports an occurrence of 312.
\end{lemma}
\begin{proof}
	\newcommand{\fD}{\mathcal{D}}
	At the beginning, after reading the first value $v$ and setting $h = v$, $D \gets \emptyset$, and $A = \{v\}$, all invariants trivially hold.
	
	Now consider a step of the algorithm. Let $\sigma$ be the sequence read before the step, $h, A, D$ be the stored data before the step, let $v$ be the value read in the step, let $\sigma' = \sigma \circ v$, and let $h',A',D'$ be the stored data after the step. We assume all invariants hold for $h, A, D$, and the algorithm does not report an occurrence of 312 in the step.
	\begin{description}
		\itemsep0pt
		\item \ref{inv:h}. Clearly, $h$ is updated correctly by \cref{line:update-h}.
		
		\item \ref{inv:A-exact}. If $v> h$, then $h' = v$ and \cref{line:delete-A,line:insert-A-1} set $A' = \{v\} \cup (A \setminus [v-k]) = \Val(\sigma') \setminus [h'-k]$.
		
		If $v \le h$, then $h' = h$, and \cref{line:insert-A-2} ensures that the invariant continues to hold.
		
		\item \ref{inv:A-no-occ}. Suppose there is an occurrence $a,b,c$ of 312 in $\tau$ such that $a > b > h'-k$ and $a$, $b$ are in~$\sigma'$. If $a,b$ are already in $\sigma$, then \ref{inv:A-no-occ} was violated previously, since $a > b > h'-k \ge h-k$, a contradiction. Thus, we have $b = v$.
		
		Since $h \ge a > b > h'-k = h-k$, the condition in \cref{line:if-into-A} is true. Since \ref{inv:A-exact} holds in the previous step, we have $a \in A$. We also clearly have $c \notin A$, and since $a > c > b$, the algorithm reports a occurrence of 312, a contradiction.
		
		\item \ref{inv:D-corr-high}. The only time a new pair is inserted into $D$ is when $(h,v)$ is inserted in \cref{line:insert-D}. Here, we have $h \ge v+k$, and $h$ must have occurred in $\sigma$ (and thus before $v$ in $\tau$) since \ref{inv:h} holds for $h$.
		
		\item \ref{inv:D-disjoint}. Say we insert $(h,v)$ in \cref{line:insert-D} and let $(a,b) \in D$ with $[b,a] \cap [v,h] \neq \emptyset$. Since $a \le h$, this implies $a > v$. If $b < v$, then we would have stopped already in \cref{line:occ-D}. If $b > v$, then $(a,b)$ is removed in \cref{line:delete-D}.
		
		\item \ref{inv:D-max}. Let $a,b,c$ be an occurrence of 312 in $\tau$ such that $a,b$ are in $\sigma'$. Define $\fD = \bigcup_{(a',b') \in D} [b',a']$, and $\fD' = \bigcup_{(a',b') \in D'} [b',a']$. By \ref{inv:D-disjoint}, it suffices to show that $[b,a] \subseteq \fD'$.
		
		First note that $\fD \subseteq \fD'$, since if we remove a pair $(a',b')$ from $D$ in \cref{line:delete-D}, then we have $h \ge a' > b' > v$ and $(h,v) \in D'$.
		
		If $b \neq v$, then $a,b$ is in $\sigma$ and $[b,a] \in \fD \subseteq \fD'$ since \ref{inv:D-max} holds in the previous step.
		
		Now suppose $b = v$. Note that $h \ge a > v$.
		If $v > h - k$, then also $a > h-k$, so $a \in A$ and $c \notin A$ by \ref{inv:A-exact}, and we report an occurrence of 312 in \cref{line:occ-A} (so we do not need to maintain the invariants after this step). If $v \le h-k$, then $\fD' = \fD \cup [b,h] \supseteq [b,a]$.\qedhere
	\end{description}
\end{proof}

\paragraph{Space usage.}
We now show that the algorithm requires $\fO(\sqrt{n \log n})$ bits of space.
Clearly, storing~$h$ requires $\fO( \log n )$ space. We can store $A$ as a bit-array of size $k$, using $\fO(k)$ space. By \ref{inv:D-disjoint}, there are at most $n/k$ pairs in $D$, so $D$ requires $\fO(\frac{n}{k} \log n)$ bits of space. The total space is thus $\fO(k + \frac{n}{k} \log n) = \fO( \sqrt{n \log n} )$.

\paragraph{Correctness.}
We show that if the algorithm reports an occurrence, then $\tau$ contains the pattern 312.
Let $\sigma$ be the prefix of $\tau$ read so far, and let $v$ be the value currently read.
First, suppose the algorithm reports an occurrence $(a,b,v)$ in \cref{line:occ-D}. By \ref{inv:D-corr-high}, $(a, b)$ is a decreasing pair in $\sigma$, and we have $a > v > b$ from \cref{line:check-D}. Thus, $a,b,v$ is an occurrence of 312 in $\tau$.

Second, suppose the algorithm reports an occurrence $(a,v,c)$ in \cref{line:occ-A}. By \ref{inv:A-exact}, we know that $a$ occurs in $\sigma$, and $c$ does not occur in $\sigma$. That means that $c$ must occur in $\tau$ somewhere after $v$, implying that $a,v,c$ is an occurrence of 312 in $\tau$.

\paragraph{Completeness.}
We now show that if $\tau$ contains the pattern 312, then the algorithm reports an occurrence.
Let $a,b,c$ be the values of an occurrence of 312 in $\tau$. Consider the step where $c$ is read, and suppose no occurrence of 312 was found yet. Let $h,A,D$ be the stored data at the beginning of the step. By invariant \ref{inv:D-max}, there is a pair $(a',b') \in D$ such that $a' \ge a > c > b \ge b'$. Then the algorithm reports the occurrence $(a', b', c)$ in \cref{line:occ-D}.

\subsection{An algorithm for 213-PPM}\label{sec:ppm:algo-213}

We now give a $\fO(\sqrt{n} \log n)$-space algorithm for 231-$\PPM_n$.
Note that this algorithm needs a factor of $\fO(\sqrt{\log n})$ more space than the algorithm from \cref{sec:ppm:algo-312}, and it does not return an occurrence when it accepts.

Let $\tau$ be the input permutation, and $n = |\tau|$. In the following, we assume that instead of receiving a value $v$, we receive a \emph{point} $(i,v) \in [n]^2$, where $i$ is the index of $x$ in $\tau$. This can be easily simulated by keeping a counter of the number of elements read so far, with $\fO(\log n)$ additive space overhead. When $p = (i,v)$, we write $p.x = i$ and $p.y = v$. Let $S$ be the set of all points induced by $\tau$ this way.

\paragraph{Algorithm.}

We partition the input into at most $\ceil{\sqrt n}$ contiguous subsequences of size at most~$\floor{\sqrt{n}}$. We call these subsequences \emph{strips}, and denote the set of points of the $i$-th strip by $S_i$.
We always store all points in the current strip at once, as well as some $\fO(\log n)$-bit data about each of the previous strips. We consider four different types of occurrences of 213, each differing by how the three points are distributed over one or more strips.
The algorithm has four parts that run in parallel, each of which detects one specific type of occurrence.
\begin{enumerate}[(1)]
	\item When we finished reading a strip $S_i$, we first search for 213 in $S_i$ by brute force. This finds all occurrences within a any single strip.
	
	\item To find occurrences $(p,q,r)$ of 213 where $p, q \in S_i$ and $r \notin S_i$, we maintain a single value $L \in [n]$, initially set to $n$. After we finished reading $S_i$, we find the lowest point $p \in S_i$ such that $p,q$ are a decreasing pair for some $q \in S_i$, i.e. $p.x < q.x$ and $q.y < p.y$. If there is such a $p$ and $p.y < L$, we update $L \leftarrow p.y$. Whenever we read a point $p$, if $p.y > L$, we accept immediately.
	
	\item To find occurrences $(p,q,r)$ with $q, r \in S_i$ for some $i$, we use a more sophisticated technique. We will need to store two points $\ell_i, h_i \in [n]^2 \cup \{\bot\}$ and a counter $k_i \in \{0, 1, \dots, h_i - \ell_i - 1\}$ for each strip $S_i$.
	
	Suppose we just finished reading $S_i$. We first define (but do not store) the set $I_i$ of increasing pairs in $S_i$ where some point in a different strip lies vertically between the two points. Formally:
	\begin{align*}
		I_i = \{ (p,q) \in S_i^2 \mid p.x < q.x, \exists r \in S \setminus S_i: p.y < r.y < q.y \},
	\end{align*}
	
	Observe that we can check if $(p,q) \in I_i$ using only $S_i$, as we know exactly which values occur in $S_i$, and, therefore, which occur in the other strips. We can also enumerate the pairs in $I_i$ one-by-one using $\fO(\log n)$ additional space.
	
	If $I_i = \emptyset$, we set $\ell_i = h_i = k_i = \bot$ and are done with this strip. Otherwise, we set $\ell_i$ to be the lowest point such that there is a pair $(\ell_i, q) \in I_i$ and set $h_i$ to be the highest point such that there is a pair $(p, h_i) \in I_i$.
	We initialize $k_i$ to be the number of points $p \in S_i$ where $\ell_i.y < p.y < h_i.y$. From now on, whenever we read a point $p$ with $\ell_i.y < p.y < h_i.y$ in a later strip, we increment $k_i$ by one. Finally, after reading the whole input, we accept if $k_i < h_i - \ell_i - 1$ (recall that we do this for each $i \in [n]$).

	\item Finally, we consider occurrences where all three points are in different strips. We use a ``counter'' technique similar to (3). After reading $S_i$, we compute the lowest point $\ell'_i$.
	Then, we initialize $h'_i$ to be the highest point in $S_i$ to the right of $\ell'_i$ and initialize the counter $k'_i$ to be the number points in $S_i$ above and to the right of $\ell'_i$. In subsequent strips, whenever we read a point $p$ with $\ell'_i.y < p.y$, we increment $k'_i$ by one. Then, if $h'_i.y < p.y$, we set $h'_i \gets p$. Finally, after reading the whole input, we accept if $k'_i < h'_i.y - \ell'_i.y$.
\end{enumerate}

If, at the end of the input, we have not accepted yet due to (1-4), we reject.

\paragraph{Space complexity.}

For each strip $S_i$, we need to store the points $\ell_i$, $h_i$, $\ell'_i$ and $h'_i$ and the values $k_i, k'_i \in [n]$. Additionally, we store $L \in [n]$ and the $\fO(\sqrt{n})$ points in $S_i$. The other operations such as the enumeration of increasing/decreasing pairs only need $\fO(\log n)$ space. Thus, the algorithm requires $\fO(\sqrt{n} \log n)$ space.

\bigskip\noindent
For the proofs of correctness and completeness, we write $L^{(i)}$ for the value of $L$ after processing strip~$S_i$. On the other hand, we only consider the final values of $k_i$ and $h'_i$.

\paragraph{Correctness.}
We show that if the algorithm accepts, then $\tau$ contains 213.

Suppose the algorithm accepts the input $\tau$. If this was due to finding 213 in a single strip via part~(1), then $\tau$ clearly contains 213. If part (2) finds a point $r \in S_i$ with $r.y < L^{(i-1)}$, then, for some $j < i$, the strip $S_j$ contains a decreasing pair of points $(p, q)$ with $p.y = L^{(i-1)} < r.y$, thus $p, q, r$ form an occurrence of 213 in $\tau$.

Now consider the case that (3) accepts, i.e. we have $k_i < h_i - \ell_i - 1$ for some $i \in [n]$ in the end. Then $k_i$ is the number of points $q$ in the strips $S_i, S_{i+1}, \dots, S_n$ with $\ell_i.y < q.y < h_i.y$. As there are $h_i.y - \ell_i.y - 1$ such points in the input, and $k_i < h_i.y - \ell_i.y - 1$, there must be at least one such point $p \in S_j$ for some $j < i$.

By definition, there are $q_\ell, p_h \in S_i$ such that $(\ell_i, q_\ell), (p_h, h_i) \in I_i$. If $\ell_i.x < h_i.x$, then $p, \ell_i, h_i$ form an occurrence of 213 in $\tau$ (see \cref{fig:streaming:213_algo_correctness}, left). Otherwise, we have $p_h.x < h_i.x < \ell_i.x$. By definition of~$I_i$, there is some input point $s \notin S_i$ with $p_h.y < s.y < h_i.y$ (see \cref{fig:streaming:213_algo_correctness}, right). If $s \in S_j$ with $j < i$, then $s, p_h, h_i$ form an occurrence of 213. On the other hand, if $s \in S_j$ for some $i < j$, then $p_h, \ell_i, s$ form an occurrence of 213, as $\ell_i.y < p_h.y$ by minimality of $\ell_i.y$.%

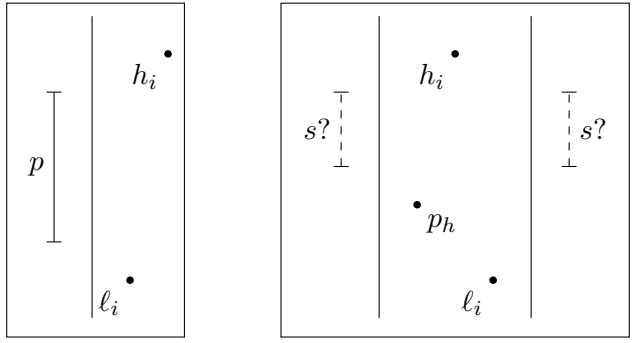
\begin{figure}
	\centering
	\begin{tikzpicture}[
		scale=0.5,
		framed,
		point/.style={fill, circle, inner sep = 1pt}
		]
		\draw[|-|] (1,2) -- (1,6);
		\node[left] at (1,4) {$p$};
		
		\draw (2,0) -- (2,8);
		
		\node[point] at (3,1) {};
		\node[below left] at (3,1) {$\ell_i$};
		
		\node[point] at (4,7) {};
		\node[below left] at (4,7) {$h_i$};
	\end{tikzpicture}
	\hspace{10mm}
	\begin{tikzpicture}[
		scale=0.5,
		framed,
		point/.style={fill, circle, inner sep = 1pt}
		]
		\draw[|-|, dashed] (1,4) -- (1,6);
		\node[left] at (1,5) {$s$?};
		
		\draw (2,0) -- (2,8);
		
		\node[point] at (3,3) {};
		\node[below right] at (3,3) {$p_h$};
		
		\node[point] at (4,7) {};
		\node[below left] at (4,7) {$h_i$};
		
		\node[point] at (5,1) {};
		\node[below left] at (5,1) {$\ell_i$};

		\draw (6,0) -- (6,8);
		
		\draw[|-|, dashed] (7,4) -- (7,6);
		\node[right] at (7,5) {$s$?};
	\end{tikzpicture}
	\caption{Illustration of the correctness of part (3) of the algorithm.}
	\label{fig:streaming:213_algo_correctness}
\end{figure}

Finally, suppose (4) accepts, i.e. $k'_i < h'_i.y - \ell'_i.y$ holds for some $i \in [n]$ in the end. $k_i'$ is the number of points to the right and above $\ell'_i$ (including $h'_i$). There are exactly $h'_i.y - \ell'_i.y - 1$ input points $p$ with $\ell'_i.y < p.y < h'_i.y$, but only $k'_i - 1$ such points that are also to the right of $\ell'_i$, so there must be at least one point $q$ with $q.x < \ell'_i.x < h'_i.x$ and $\ell'_i.y < q.y < h'_i.y$. Thus $(q, \ell'_i, h'_i)$ form an occurrence of 213 in~$\tau$.

\paragraph{Completeness.}
We show that if $\tau$ contains 213, then the algorithm accepts.

Suppose that the input contains three points $p, q, r$ that form an occurrence of 213. If there is some $i$ such that $p, q, r \in S_i$, part (1) will find $(p,q,r)$.

Suppose $p, q \in S_i$ and $r \notin S_i$ for some $i$. We show that then (2) accepts. Let $(p', q')$ be a decreasing pair in $S_i$ such that $p'.y$ is minimal. Let $r \in S_j$. As $i \le (j-1)$, by definition of $L$ and minimality of $p'.y$, we know $L^{(j-1)} \le p'.y \le p.y < r.y$. As such, (2) accepts when the algorithm reads $r$ and notices $L < r.y$.

Suppose $p \notin S_i$ and $q, r \in S_i$ for some $i$. We show that then (3) accepts. We have $\ell_i.y \le q.y < p.y < r.y \le h_i.y$. As $p \in S_j$ for some $j < i$, the number $k_i$ of points $s$ with $\ell_i.y < s.y < h_i.y$ to the right of $S_i$ is less than the $h_i.y - \ell_i.y - 1$ (the total number of such input points). Thus, (3) accepts.

Finally, suppose $p$, $q$ and $r$ are in different strips. We show that then (4) accepts. Let $q \in S_i$. Then $\ell'_i.y < q.y < p.y < r.y < h'_i.y$. As $p \in S_j$ for some $j < i$, there the number $k'_i$ of points $s$ to the right of $\ell'_i$ which satisfy $\ell'_i.y < s.y$ (and, by definition, $s.y \le h'_i.y$) is at most $h'_i.y - \ell'_i.y - 1$. Thus, we accept.

\subsection{Two-party communication complexity}\label{sec:ppm:3-cc}

We now show that any non-constant lower bound for $\pi$-$\PPM_n$ with $\pi \in \{312, 132, 213, 231\}$ is impossible to prove with a straight-forward reduction from two-party communication complexity. Define the \emph{one-way two-party $\pi$-$\PPM_n$} problem as follows. The input is a permutation $\tau$ of~$[n]$, as usual. We have two players, Alice and Bob. Alice receives a prefix of $\tau$ (of any length), and Bob receives the remaining suffix.
The \emph{communication complexity} of one-way two-party $\pi$-$\PPM_n$ is the minimum number of bits of information that Alice needs to send to Bob so that Bob can decide whether $\pi$ is contained in $\tau$.

Clearly, this provides a lower bound for the space complexity of $\pi$-$\PPM_n$, and potentially a stronger one than the \emph{two-way} communication complexity used in the rest of the paper. Unfortunately:

\begin{proposition}\label{p:cc-ppm-easy}
	The communication complexity of one-way two-party $\pi$-$\PPM_n$ is one bit if $|\pi| \le 3$.
\end{proposition}
\begin{proof}
	Let $\alpha, \beta$ be the two parts of the input received by Alice and Bob.
	Alice can detect occurrences of $\pi$ that are fully contained in $\alpha$, and, since she knows $\Val(\beta)$, she can also detect occurrences of $\pi$ where all but the last value are contained in $\alpha$. Similarly, Bob can detect all occurrences where all values or all but the first value are contained in $\beta$. If $|\pi| \le 3$, this encompasses all possible occurrences, and hence Alice only needs to send one bit indicating whether she found an occurrence.
\end{proof}

\section{Conclusion}\label{sec:conc}

In this paper we initiated the study of pattern matching in permutation streams. We showed hardness of the problem for almost all non-monotone patterns, but our $\widetilde{\fO}(\sqrt{n})$-space algorithms for the patterns 312, 132, 213, and 231 show a separation from the previously studied \emph{sequence} variant, where a $\Omega(n)$ lower bound holds. We conjecture that our upper bound is tight, which would imply a surprising \emph{trichotomy} of complexities $\Theta(\log n)$, $\widetilde{\Theta}(\sqrt{n})$, and $\widetilde{\Theta}(n)$.

\begin{conjecture}
	The space complexity of $\pi$-$\PPM_n$ is $\widetilde{\Theta}(\sqrt{n})$ for $\pi \in \{312, 132, 213, 231\}$.
\end{conjecture}

As discussed in \cref{sec:ppm:3-cc}, such a lower bound likely requires some more sophisticated techniques.

It is also interesting that the two algorithms for detecting 312/132 and 213/231 are quite different, even though the patterns are the reverses of each other. Both essentially rely on the fact that each element is known to occur exactly once, in order to circumvent the $\Omega(n)$ bound for $\SeqPPM$. In the first algorithm, this is exploited ``early'': At a certain point (\cref{line:occ-A} in \cref{alg:312}), the algorithm may deduce that a certain element must appear later, and stop right before its internal storage is overwhelmed by the complexity of the input. The second algorithm instead counts elements in certain ranges, and \emph{at the end} deduces that some values must have occurred before the counting started. This means the second algorithm cannot always return an actual occurrence, but only decide whether such an occurrence exists.
We do not know if an occurrence of 213 (or 231) can be \emph{found} with $\widetilde{\Theta}(\sqrt{n})$ space; this is another interesting open question.

As mentioned in the introduction, \emph{counting} patterns is already hard for the pattern $\pi = 21$ (i.e., counting inversions)~\cite{AjtaiEtAl2002}. However, this is not known for counting \emph{modulo 2}; in the case $\pi = 21$, this means computing the \emph{parity} of the input. Reductions from communication complexity are faced with a similar problem as for $\PPM$ with length-3 patterns (\cref{sec:ppm:3-cc}): Two players can compute the parity of a permutation with only one bit of communication.\footnote{This was observed by Dömötör P\'{a}lvölgyi at \url{https://cstheory.stackexchange.com/a/27912} (accessed 2025/7/14).}

Note that counting inversions \emph{modulo 3} can be shown to be hard with a simple modification of the Ajtai--Jayram--Kumar--Sivakumar lower bound\footnote{The \emph{disjointness} problem in the proof needs to be replaced by the \emph{unique disjointness problem}~\cite{Razborov1990UniqueDisjointness}.}~\cite{AjtaiEtAl2002}. Also, computing the parity is hard in the \emph{sequence} setting, as can be shown by a simple reduction from the \CCIndex{} problem.

Our proofs imply that counting some larger patterns is hard even modulo 2: The $\CCDisj$-reductions for size-4 patterns (\cref{sec:disj-hardness}) all have the property that precisely one occurrence appears for each element in the intersection of the two input sets. Thus, by reducing from the \emph{unique disjointness} problem~\cite{Razborov1990UniqueDisjointness} instead, we can show that distinguishing between zero and one occurrence is hard.

Finally, it is known how to \emph{approximately} count inversions efficiently~\cite{AjtaiEtAl2002}; we do not know whether this is possible for larger patterns.

\paragraph{Acknowledgments.}
Many thanks to László Kozma, who initially suggested the topic for the author's Master's thesis~\cite{Berendsohn2019}, and contributed many of the initial ideas. Additional thanks to Paweł Gawrychowski for suggesting the $\sqrt{\log n}$-factor improvement to the 312-$\PPM$ algorithm, and to Michal Opler, Marek Sokołowski, and Or Zamir for useful discussions around this topic throughout the years.

\bibliography{main}
\bibliographystyle{alphaurl}

\end{document}